\documentclass[12pt]{article}

\usepackage{fullpage}
\usepackage{amsthm}
\usepackage{amsmath}
\usepackage{amssymb}
\usepackage{mathtools}
\usepackage{enumerate}
\usepackage{thm-restate}

\usepackage[usenames,dvipsnames]{pstricks}
\usepackage{epsfig}

\newif\iftikz
\tikztrue

\iftikz
\usepackage{tikz}

\fi

\newtheorem{theorem}{Theorem} 
\newtheorem{corollary}[theorem]{Corollary}

\newtheorem{lemma}[theorem]{Lemma}

\newtheorem{theorem*}{Theorem}

\theoremstyle{definition}
\newtheorem{definition}[theorem]{Definition}

\IfFileExists{../EmanueleViolaDir.txt}{\def\EmanueleViolaDir{1}}{
\def\EmanueleViolaDir{0}}

\ifnum\EmanueleViolaDir=1
\usepackage{../theomac}
\usepackage{../bibspacing}
\else
\usepackage{theomac}
\usepackage{bibspacing}
\fi

\newtheoremWithMacro{theoremR}[theorem]{Theorem}
\newtheoremWithMacro{lemmaR}[theorem]{Lemma}
\newtheoremWithMacro{factR}[theorem]{Fact}

\newcommand{\poly}{\mathrm{poly}}

\newcommand{\zo}{\{0, 1\}}

\newcommand{\e}{\epsilon}
\newcommand{\eps}{\epsilon}

\newcommand{\acz}{\mathrm{AC}^0}
\newcommand{\ncz}{\mathrm{NC}^0}

\begin{document}

\begin{titlepage}
\title{Local reductions}
\author{Hamid Jahanjou\thanks{Supported by NSF grants CCF-0845003, CCF-1319206. Email:
\texttt{\{hamid,enmiles,viola\}@ccs.neu.edu}} \and Eric
Miles\footnotemark[1] \and Emanuele
Viola\footnotemark[1]}

\maketitle

\begin{abstract}
We reduce non-deterministic time $T \ge 2^n$ to a 3SAT
instance $\phi$ of quasilinear size $|\phi| = T \cdot
\log^{O(1)} T$ such that there is an explicit circuit $C$
that on input an index $i$ of $\log |\phi|$ bits outputs
the $i$th clause, and each output bit of $C$ depends on
$O(1)$ input bits. The previous best result was $C$ in
NC$^1$. Even in the simpler setting of polynomial size
$|\phi| = \poly(T)$ the previous best result was $C$ in
AC$^0$.

More generally, for any time $T \ge n$ and parameter $r
\leq n$ we obtain $\log_2 |\phi| = \max(\log T, n/r) +
O(\log n) + O(\log\log T)$ and each output bit of $C$ is
a decision tree of depth $O(\log r)$.

As an application, we tighten Williams' connection
between satisfiability algorithms and circuit lower
bounds (STOC 2010; SIAM J.~Comput.~2013).
\end{abstract}

\thispagestyle{empty}

\end{titlepage}

\newpage

\section{Introduction and our results} \label{s-intro}

The efficient reduction of arbitrary non-deterministic
computation to 3SAT is a fundamental result with
widespread applications.  For many of these, two aspects
of the efficiency of the reduction are at a premium.  The
first is the length of the 3SAT instance.  A sequence of
works shows how to reduce non-deterministic time-$T$
computation to a 3SAT instance $\phi$ of quasilinear size
$|\phi| = \tilde O(T) := T \log^{O(1)} T$
\cite{HennieS66,Schnorr78,PippengerF79,Cook88,GurevichS89,Robson91}.
This has been extended to PCP reductions
\cite{Ben-SassonGHSV05,Mie09,Ben-SassonCGT12Fast,Ben-SassonCGT12On}.

The second aspect is the computational complexity of
producing the 3SAT instance $\phi$ given a machine $M$,
an input $x \in \zo^n$, and a time bound $T = T(n) \ge
n$. It is well-known and easy to verify that a $\phi$ of
size $\poly(T)$ is computable even by circuits from the
restricted class NC$^0$.
More generally, Agrawal, Allender, Impagliazzo, Pitassi,
and Rudich show \cite{AAIPR01} that such NC$^0$
reductions exist whenever AC$^0$ reductions do.


A stronger requirement on the complexity of producing
$\phi$ is critical for many applications.  The
requirement may be called \emph{clause-explicitness}.  It
demands that the $i$th clause of $\phi$ be computable,
given $i \le |\phi|$ and $x \in \zo^n$, with resources
$\poly(|i|) = \poly \log |\phi| = \poly \log T$.  In the
case $|\phi| = \poly(T)$, this is known to be possible by
an unrestricted circuit $D$ of size $\poly(|i|)$.  (The
circuit has either random access to $x$, or, if $T \ge
2^n$, it may have $x$ hardwired.) As a corollary,
so-called succinct versions of NP-complete problems are
complete for NEXP.
Arora, Steurer, and Wigderson \cite{AroraSW09} note that
the circuit $D$ may be taken from the restricted class
AC$^0$.
They use this to argue that, unless
EXP = NEXP, standard NP-complete graph problems cannot be
solved in time $\poly(2^n)$ on graphs of size $2^n$ that
are described by AC$^0$ circuits of size $\poly(n)$.

\medskip

Interestingly, applications to unconditional complexity
lower bounds rely on reductions that are clause-explicit
and simultaneously optimize the length of the 3SAT
instance $\phi$ and the complexity of the circuit $D$
computing clauses. For example, the time-space tradeoffs
for SAT need to reduce non-deterministic time $T$ to a
3SAT instance $\phi$ of quasilinear size $\tilde O(T)$
such that the $i$th clause is computable in time
$\poly(|i|) = \poly \log |\phi|$ and space $O(\log
|\phi|)$, see e.g.~\cite{FLvMV05} or Van Melkebeek's
survey \cite{Melkebeek06}.  More recently, the importance
of optimizing both aspects of the reduction is brought to
the forefront by Williams' approach to obtain lower
bounds by satisfiability algorithms that improve over
brute-force search by a super-polynomial factor
\cite{Williams13-Improving,Williams10acc,Williams11,SanthanamW12,Williams13}.
To obtain lower bounds against a circuit class $C$ using
this technique, one needs a reduction of
non-deterministic time $T = 2^n$ to a 3SAT instance of
size $\tilde O(T)$ whose clauses are computable by a
circuit $D$ of size $\poly(n)$ that belongs to the class
$C$.  For example, for the ACC$^0$ lower bounds
\cite{Williams10acc,Williams13} one needs to compute them
in ACC$^0$. However it has seemed ``hard (perhaps
impossible)'' \cite{Williams10acc} to compute the clauses
with such restricted resources.

Two workarounds have been devised
\cite{Williams10acc,SanthanamW12}. Both exploit the fact
that, under an assumption such as P $\subseteq$ ACC$^0$,
non-constructively there does exist such an efficient
circuit computing clauses; the only problem is
constructing it.  They accomplish the latter by
guessing-and-verifying it \cite{Williams10acc}, or by
brute-forcing it \cite{SanthanamW12}
(cf.~\cite{AllenderK10}).  The overhead in these
arguments limits the consequences of satisfiability
algorithms: before this work, for a number of
well-studied circuit classes $C$ (discussed later) a
lower bound against $C$ did not follow from a
satisfiability algorithm for circuits in $C$.

\subsection{Our results} \label{s-our-results}

We show that, in fact, it is possible to reduce
non-deterministic computation of time $T \ge 2^n$ to a
3SAT formula $\phi$ of quasilinear size $|\phi| = \tilde
O(T)$ such that given an index of $\ell = \log |\phi|$
bits to a clause, one can compute (each bit of) the
clause by looking at a constant number of bits of the
index. Such maps are also known as local, NC$^0$, or junta.
More generally our results give a trade-off between
decision-tree depth and $|\phi|$.
The results apply to any time bound $T$,
paying an inevitable loss in $|x| = n$ for $T$ close to
$n$.

\begin{restatable}[Local reductions]{theorem}{texplicitreductions} \label{t-explicit-reductions}
Let $M$ be an algorithm running in time $T = T(n)\geq n$ on
inputs of the form $(x,y)$ where $|x| = n$.  Given $x \in
\zo^n$ one can output a circuit $D : \zo^\ell \to \zo^{3v
+ 3}$ in time $\poly(n, \log T)$ mapping an index to a clause of a 3CNF $\phi$ in
$v$-bit variables, for $v = \Theta(\ell)$, such that
\begin{enumerate}
\item $\phi$ is satisfiable iff there is $y \in
\zo^T$ such that $M(x,y) = 1$, and
\item for any $r \leq n$ we can have $\ell = \max(\log T, n/r) + O(\log n) + O(\log\log T)$ and each output bit of $D$ is a decision tree of depth $O(\log r)$.
\end{enumerate}
\end{restatable}

Note that for $T= 2^{\Omega(n)}$ we get that $D$ is in
$\ncz$ and $\phi$ has size $2^\ell = T \cdot \log^{O(1)}
T$, by setting $r := n / \log T$.  We also point out that
the only place where locality $O(\log r)$ (as opposed to
$O(1)$) is needed in $D$ is to index bits of the string
$x$.

%
%
%
%


The previous best result was $D$ in NC$^1$
\cite{Ben-SassonGHSV05}.  Even in the simpler setting of
$|\phi| = \poly(T)$ the previous best result was $D$ in
AC$^0$ \cite{AroraSW09}.



\paragraph{Tighter connections between satisfiability and
lower bounds.}  The quest for non-trivial satisfiability
algorithms has seen significant progress recently, see
e.g.~\cite{Williams10acc,Hertli-3SAT,ImpagliazzoMP-AC0-sat,BeameIS-AC0-SAT,ImpagliazzoPS-SAT-sparse,ChenKS-SAT-formula}.
Our results lower the bar for obtaining new circuit lower
bounds from such algorithms.
Previously, a lower bound for circuits of depth $d$ and
size $s$ was implied by a satisfiability algorithm for
depth $c \cdot d$ and size $s^c$ for a constant $c > 1$
(for typical settings of $s$ and $d$).  With our proof it
suffices to have a satisfiability algorithm for depth $d + c$
and size $c \cdot s$ for a constant $c$.  These
results can be extended and optimized for several
well-studied circuit classes.  In particular we obtain
the following new connections.

\begin{corollary} \label{co-sat-c-lb-c}
For each of the following classes $C$, if the satisfiability of circuits in $C$ can be solved in time $2^n/n^{\omega(1)}$ then there is a problem $f \in$ {\em E}$^\mathrm{NP}$ that is not solvable by circuits in $C$:

(1) linear-size circuits,

(2) linear-size series-parallel circuits,

(3) linear-size log-depth circuits,

(4) quasi-polynomial-size SYM-AND circuits.
\end{corollary}

Recall that available size lower bounds for unrestricted
circuits are between $3n-o(n)$ and $5n-o(n)$, depending on the
basis \cite{Blum84,LachishR01,IwamaM02}.
In 1977 Valiant \cite{Val77} focused attention on classes
(2) and (3).  (Some missing details about series-parallel
graphs are provided in \cite{Calabro08}.)  The class (4)
contains ACC \cite{Yao90,BeT94}, and can be
simulated by number-on-forehead protocols with
a polylogarithmic number of players and communication \cite{HaG91}.
Williams \cite{Williams10acc} gives a quasilinear-time
algorithm to evaluate a SYM-AND circuit on all inputs.

For class (4) one can in fact obtain $f \in$ NE using the
seminal work by Impagliazzo, Kabanets, and Wigderson
\cite{ImpagliazzoKaWi01} and its extension by Williams
\cite{Williams13-Improving,Williams10acc}.  But to do so
for classes (1)-(3), one would need a strengthening of
\cite{ImpagliazzoKaWi01} to linear-size circuits, which
we raise as an open problem.

It has long been known that the satisfiability of classes
(1)-(3) can be linked to $k$SAT.  Using Corollary
\ref{co-sat-c-lb-c}, we can link $k$SAT to circuit lower
bounds.
Only (1') in the following corollary was known
\cite[Theorem 6.1]{Williams13-Improving}.  Our proof is
different: we obtain (1') immediately from (1) in
Corollary \ref{co-sat-c-lb-c} by the Cook-Levin theorem.
In the corollary, a $k$SAT instance has $n$ variables and
$O(n)^k$ clauses.

\begin{corollary} \label{co-CNF-SAT-lb}
In Items (1), (2), and (3) of Corollary
\ref{co-sat-c-lb-c} we can substitute, respectively, the
following assumptions:

(1') The exponential time hypothesis (ETH)
\cite{ImpagliazzoP01} is false; i.e., for every $\eps
> 0$, 3SAT
is in time $2^{\eps n}$,

(2') The strong exponential time hypothesis (SETH) is
false \cite{ImpagliazzoP01}; i.e., there is $\eps
< 1$ such that for every $k$, $k$SAT
is in time $2^{\eps n}$,

(3') there is $\alpha > 0$ such that $n^{\alpha}$-SAT
is in time $2^{n - \omega(n/\log\log n)}$.
\end{corollary}

For context, the best algorithms for $k$SAT run in time
$2^{n(1 - O(1/k))}$ \cite{DantsinDHKKPRS02,PPSZ05}.

Although Corollaries \ref{co-sat-c-lb-c} and
\ref{co-CNF-SAT-lb} are stated in terms of linear-size
circuits, the proofs provide a close correspondence
between the running time for satisfiability and the
parameters of the circuit class.  This is discussed in
more detail later in this section, see ``subsequent
work.''

Finally, we consider the class of polynomial-size
depth-$d$ circuits of threshold gates, which may have
unbounded or bounded weights.  (The latter case
corresponds to Majority.)  Recall that anything computed
by a poly-size depth-$d$ circuit with unbounded weights
can be computed by a depth $d+1$ circuit with bounded
weights \cite{HMPST93,GoldmannHR92}, and that it is not
known if EXP$^\mathrm{NP}$ has poly-size unbounded-weight
circuits of depth $d=2$.  For these classes (and others)
we show that a lower bound for depth $d$ follows from a
satisfiability algorithm for depth $d+2$.

\begin{corollary} \label{co-threshold-depth} Consider unbounded fan-in
circuits consisting of threshold gates (either bounded-
or unbounded-weight).  Let $d$ be an integer.

Suppose that for every $c$, given a circuit of depth
$d+2$ and size $n^c$ on $n$ input bits one can decide its
satisfiability in time $2^n/n^{\omega(1)}$.

Then NE does not have circuits of polynomial size and
depth $d$.
\end{corollary}

A diagram of some of the classes mentioned above, and
their relative power, can be found in \cite{Viola-map}.

\medskip

Our results have a few other consequences.  For example they imply that the so-called succinct version of various NP-complete problems remain NEXP-complete even if described by an $\ncz$ circuit. In particular we obtain this for 3SAT and 3Coloring.
Our techniques are also relevant to the notion of circuit
uniformity.  A standard notion of uniformity is log-space
uniformity, requiring that the circuit is computable in
logarithmic space or, equivalently, that given an index to a
gate in the circuit one can compute its type and its
children in linear space.  Equivalences with various
other uniformity conditions are given by Ruzzo
\cite{Ruzzo81}, see also \cite{Vol99}.  We consider
another uniformity condition which is stronger than
previously considered ones in some respects.
Specifically, we describe the circuit by showing how to
compute children by an NC$^0$ circuit, i.e.\ a function
with constant locality.

\begin{restatable}[L-uniform $\Leftrightarrow$ local-uniform]{theorem}{tlocaluniformity}
\label{t-local-uniformity} Let $f : \zo^* \to \zo$ be a
function computable by a family of log-space uniform
polynomial-size circuits.  Then $f$ is computable by a
family of polynomial-size circuits $C = \{C_n : \zo^n \to
\zo\}_n$ such that there is Turing machine that on input
$n$ (in binary) runs in time $O(\poly \log n)$ and
outputs a map
$D : \zo^{O(\log n)} \to \zo^{O(\log n)}$ such that\\
(i) $D$ has constant locality, i.e., every output bit of
$D$ depends on $O(1)$ input bits, and \\
(ii) on input a label $g$ of a gate in $C_n$, $D$ outputs
the type of $g$ and labels for each child.
\end{restatable}


\paragraph{Does this paper simplify the proof that NEXP is not in
ACC?} Recall that the proof \cite{Williams10acc} that
NEXP is not in ACC uses as a black-box a result like
Theorem \ref{t-explicit-reductions} but with the
requirement on the efficiency of $D$ relaxed to
polynomial-size circuits.  If one instead uses as
a black-box Theorem \ref{t-explicit-reductions}, one
obtains an arguably more direct proof, reported for
completeness in \S\ref{s-sat-2-lb}.

In fact, to obtain the separation of NEXP from ACC it
suffices to prove a weaker version of Theorem
\ref{t-explicit-reductions} where $D$ is, say, in AC$^0$.
This weaker version has a simpler proof, as explained in
\S\ref{s-intro-techniques}. Independently of our work,
Kowalski and Van Melkebeek proved this AC$^0$ result
(personal communication).

\paragraph{Subsequent work.} The
announcement of our results as (ECCC Technical Report
13-099, July 2013) contained the same results as above
except it did not mention Corollary \ref{co-CNF-SAT-lb}
and items (2) and (4) in Corollary \ref{co-sat-c-lb-c}.
After that announcement several related works have
appeared.  Oliveira's survey \cite{Oliveira13} contains
an alternative connection between satisfiability and
circuit lower bounds, which yields a different proof of
our Corollary \ref{co-threshold-depth} establishing a
depth-2 overhead in that connection.  Williams
\cite{Williams14} shows that the ability to count the
number of satisfying assignments to circuits faster than
brute-force search yields lower bounds against related
circuits.  His connection preserves the type of the gates
in the input layer, a feature which is used to obtain
some new lower bounds.

The work \cite{Ben-SassonV-AC0PCP} builds on our results
and is concurrent with \cite{Williams14}.  It gives a
connection between derandomization and lower bounds that
also preserves the type of the gates in the input layer.
Thus, derandomization (or satisfiability), as opposed to
counting, is sufficient for the lower bounds in
\cite{Williams14}. \cite{Ben-SassonV-AC0PCP} also
improves the depth loss of 2 in Corollary
\ref{co-threshold-depth} to 1.  Finally, they make a step
in the direction we suggested of optimizing the constants
in Item (1) of Corollary \ref{co-sat-c-lb-c}.  In
combination with the standard Cook-Levin reduction to
3SAT, they obtain that if 3SAT is in deterministic time
$c^n$ for any $c < 2^{1/10} = 1.07\ldots$ then E$^{NP}$
does not have circuits of size $3n$ over the standard,
full basis. Note that such a lower bound does not easily
follow from diagonalization because the description
length of a circuit of size $3n$ is superlinear.  (Also
recall the available lower bounds have the form
$3n-o(n)$).  The current record for solving 3SAT
deterministically has $c = 1.33\ldots$ \cite{MakinoTY11},
cf.~\cite{Hertli-3SAT}.

As a corollary to \cite{Ben-SassonV-AC0PCP}, in this
revision we show that even a somewhat more modest
improvement to 3SAT algorithms would imply new lower
bounds for non-boolean functions with range $m=2$ bits.
Such lower bounds do not seem known for any $m = o(n)$,
cf.~\cite{KulikovMM12}.

\begin{corollary}[Corollary to \cite{Ben-SassonV-AC0PCP}]\label{co-3sat-with-3col}
If 3SAT is in time $c^n$ for any $c <
2^{1/7}=1.10\ldots$, then there exists a (non-Boolean)
function $f:\zo^n\rightarrow \zo^2$ in E$^{NP}$ such that
any circuit over the full basis computing it requires
at least $3n$ (non-input) gates.
\end{corollary}

\subsection{Techniques} \label{s-intro-techniques}


\paragraph{Background: Reducing non-deterministic
time $T$ to size-$\tilde O(T)$ 3SAT.}  Our starting point
is the reduction of non-deterministic time-$T$
computation to 3SAT instances of quasilinear size $T' =
\tilde O(T)$.  The classical proof of this result
\cite{HennieS66,Schnorr78,PippengerF79,Cook88,GurevichS89,Robson91}
hinges on the oblivious Turing machine simulation by
Pippenger and Fischer \cite{PippengerF79}. However
computing connections in the circuit induced by the
oblivious Turing machine is a somewhat complicated
recursive procedure, and we have not been able to use
this construction for our results.

Instead, we use an alternative proof that replaces this
simulation by coupling an argument due to Gurevich and
Shelah \cite{GurevichS89} with sorting circuits.  The
first reference that we are aware of for the alternative
proof is the survey by Van Melkebeek \cite[\S
2.3.1]{Melkebeek06}, which uses Batcher's odd-even
mergesort networks \cite{Batcher68}. This proof was
rediscovered by a superset of the authors as a class
project \cite{ViolaNEU-ram2sat}. We now recall it.

Consider any general model of (non-deterministic)
computation, such as RAM or random-access Turing
machines.  (One nice feature of this proof is that it
directly handles models with random-access, aka
direct-access, capabilities.)  The proof reduces
computation to the satisfiability of a circuit $C$. The
latter is then reduced to 3SAT via the textbook
reduction.  Only the first reduction to circuit
satisfiability is problematic and we will focus on that
one here.  Consider a non-deterministic time-$T$
computation.  The proof constructs a circuit of size
$\tilde O(T)$ whose inputs are (non-deterministic guesses
of) $T$ configurations of the machine.  Each
configuration has size $O(\log T)$ and contains the state
of the machine, all registers, and the content of the
memory locations indexed by the registers. This
computation is then verified in two steps.  First, one
verifies that every configuration $C_i$ yields
configuration $C_{i+1}$ assuming that all bits read from
memory are correct. This is a simple check of adjacent
configurations.  Then to verify correctness of read/write
operations in memory, one sorts the configurations by
memory indices, and within memory indices by timestamp.
Now verification is again a simple check of adjacent
configurations.  The resulting circuit is outlined in
Figure \ref{fig:sort-and-check} (for a $2k$-tape
random-access Turing machine). Using a sorting network of
quasilinear size $\tilde O(T)$ results in a circuit of
size $\tilde O(T)$.

\begin{figure}
\scalebox{.7} 
{
\begin{pspicture}(1.5,-14.14)(25.186897,14.16)
\psframe[linewidth=0.04,dimen=outer](4.0,-13.42)(0.0,-14.14)
\usefont{T1}{ptm}{m}{n}
\rput(2.0373437,-13.765){\large $c_1$}
\psdots[dotsize=0.12](12.38,-13.74)
\psdots[dotsize=0.12](12.98,-13.74)
\psdots[dotsize=0.12](13.58,-13.74)
\usefont{T1}{ptm}{m}{n}
\rput(10.480469,-9.85){sort by $Ram_1$ head position}
\pstriangle[linewidth=0.04,dimen=outer](4.5,-13.02)(4.6,1.96)
\usefont{T1}{ptm}{m}{n}
\rput(4.5728126,-12.38){\small head positions,}
\usefont{T1}{ptm}{m}{n}
\rput(4.540625,-12.76){\small  bounded-register tapes}
\usefont{T1}{ptm}{m}{n}
\rput(4.528125,-12.0){\small check state,}
\psframe[linewidth=0.04,dimen=outer](8.98,-13.42)(4.98,-14.14)
\usefont{T1}{ptm}{m}{n}
\rput(7.0173435,-13.765){\large $c_2$}
\psframe[linewidth=0.04,dimen=outer](20.98,-13.42)(16.98,-14.14)
\usefont{T1}{ptm}{m}{n}
\rput(19.067345,-13.765){\large $c_T$}
\pstriangle[linewidth=0.04,dimen=outer](9.5,-13.02)(4.6,1.96)
\usefont{T1}{ptm}{m}{n}
\rput(9.572812,-12.38){\small head positions,}
\usefont{T1}{ptm}{m}{n}
\rput(9.540625,-12.76){\small  bounded-register tapes}
\usefont{T1}{ptm}{m}{n}
\rput(9.528125,-12.0){\small check state,}
\pstriangle[linewidth=0.04,dimen=outer](16.46,-13.02)(4.6,1.96)
\usefont{T1}{ptm}{m}{n}
\rput(16.532812,-12.38){\small head positions,}
\usefont{T1}{ptm}{m}{n}
\rput(16.500626,-12.76){\small  bounded-register tapes}
\usefont{T1}{ptm}{m}{n}
\rput(16.488125,-12.0){\small check state,}
\psline[linewidth=0.04cm](3.1,-13.46)(3.1,-13.0)
\psline[linewidth=0.04cm](5.9,-13.46)(5.9,-13.0)
\psline[linewidth=0.04cm](8.1,-13.46)(8.1,-13.0)
\psline[linewidth=0.04cm](17.88,-13.46)(17.88,-13.0)
\psline[linewidth=0.04cm](1.98,-13.44)(1.98,-10.24)
\psline[linewidth=0.04cm](6.98,-13.44)(6.98,-10.22)
\psline[linewidth=0.04cm](18.98,-13.44)(18.98,-10.24)
\psframe[linewidth=0.04,dimen=outer](4.0,-7.62)(0.0,-8.34)
\usefont{T1}{ptm}{m}{n}
\rput(2.0373437,-7.965){\large $c_1$}
\psdots[dotsize=0.12](12.4,-7.94)
\psdots[dotsize=0.12](13.0,-7.94)
\psdots[dotsize=0.12](13.6,-7.94)
\pstriangle[linewidth=0.04,dimen=outer](4.5,-7.22)(4.6,1.96)
\usefont{T1}{ptm}{m}{n}
\rput(4.494375,-6.96){\small check $Ram_1$ contents}
\psline[linewidth=0.04cm](4.5,-5.3)(4.5,-4.86)
\usefont{T1}{ptm}{m}{n}
\rput(7.1673436,-7.965){\large $c_{87}$}
\psframe[linewidth=0.04,dimen=outer](20.98,-7.62)(16.98,-8.34)
\usefont{T1}{ptm}{m}{n}
\rput(19.167343,-7.965){\large $c_{42}$}
\psline[linewidth=0.04cm](3.1,-7.66)(3.1,-7.2)
\psline[linewidth=0.04cm](5.9,-7.66)(5.9,-7.2)
\psline[linewidth=0.04cm](8.1,-7.66)(8.1,-7.2)
\psline[linewidth=0.04cm](17.88,-7.66)(17.88,-7.2)
\psline[linewidth=0.04cm](10.88,-13.46)(10.88,-13.0)
\psline[linewidth=0.04cm](15.1,-13.46)(15.1,-13.0)
\psline[linewidth=0.04cm](1.98,-9.44)(1.98,-8.32)
\psline[linewidth=0.04cm](6.98,-9.44)(6.98,-8.32)
\psline[linewidth=0.04cm](18.98,-9.46)(18.98,-8.3)
\psline[linewidth=0.04cm](4.5,-11.08)(4.5,-10.64)
\psline[linewidth=0.04cm](9.5,-11.08)(9.5,-10.64)
\psline[linewidth=0.04cm](16.46,-11.08)(16.46,-10.64)
\pstriangle[linewidth=0.04,dimen=outer](9.5,-7.22)(4.6,1.96)
\usefont{T1}{ptm}{m}{n}
\rput(9.494375,-6.96){\small check $Ram_1$ contents}
\psline[linewidth=0.04cm](9.5,-5.3)(9.5,-4.86)
\pstriangle[linewidth=0.04,dimen=outer](16.5,-7.22)(4.6,1.96)
\usefont{T1}{ptm}{m}{n}
\rput(16.494375,-6.96){\small check $Ram_1$ contents}
\psline[linewidth=0.04cm](16.5,-5.3)(16.5,-4.86)
\psline[linewidth=0.04cm](15.08,-7.66)(15.08,-7.2)
\psline[linewidth=0.04cm](10.88,-7.66)(10.88,-7.2)
\pspolygon[linewidth=0.04](0.0,-10.22)(0.4990476,-9.44)(20.460953,-9.44)(20.96,-10.22)
\usefont{T1}{ptm}{m}{n}
\rput(10.480469,-4.05){sort by $Ram_2$ head position}
\psline[linewidth=0.04cm](1.98,-7.64)(1.98,-4.44)
\psline[linewidth=0.04cm](6.98,-7.64)(6.98,-4.42)
\psline[linewidth=0.04cm](18.98,-7.64)(18.98,-4.44)
\psframe[linewidth=0.04,dimen=outer](4.0,-1.82)(0.0,-2.54)
\usefont{T1}{ptm}{m}{n}
\rput(2.0373437,-2.165){\large $c_1$}
\psdots[dotsize=0.12](12.4,-2.14)
\psdots[dotsize=0.12](13.0,-2.14)
\psdots[dotsize=0.12](13.6,-2.14)
\pstriangle[linewidth=0.04,dimen=outer](4.5,-1.42)(4.6,1.96)
\usefont{T1}{ptm}{m}{n}
\rput(4.494375,-1.16){\small check $Ram_2$ contents}
\psline[linewidth=0.04cm](4.5,0.5)(4.5,0.94)
\psframe[linewidth=0.04,dimen=outer](8.98,-1.82)(4.98,-2.54)
\usefont{T1}{ptm}{m}{n}
\rput(7.1673436,-2.165){\large $c_{19}$}
\psframe[linewidth=0.04,dimen=outer](20.98,-1.82)(16.98,-2.54)
\usefont{T1}{ptm}{m}{n}
\rput(19.167343,-2.165){\large $c_{71}$}
\psline[linewidth=0.04cm](3.1,-1.86)(3.1,-1.4)
\psline[linewidth=0.04cm](5.9,-1.86)(5.9,-1.4)
\psline[linewidth=0.04cm](8.1,-1.86)(8.1,-1.4)
\psline[linewidth=0.04cm](17.88,-1.86)(17.88,-1.4)
\psline[linewidth=0.04cm](1.98,-3.64)(1.98,-2.52)
\psline[linewidth=0.04cm](6.98,-3.64)(6.98,-2.52)
\psline[linewidth=0.04cm](18.98,-3.66)(18.98,-2.5)
\pstriangle[linewidth=0.04,dimen=outer](9.5,-1.42)(4.6,1.96)
\usefont{T1}{ptm}{m}{n}
\rput(9.494375,-1.16){\small check $Ram_2$ contents}
\psline[linewidth=0.04cm](9.5,0.5)(9.5,0.94)
\pstriangle[linewidth=0.04,dimen=outer](16.5,-1.42)(4.6,1.96)
\usefont{T1}{ptm}{m}{n}
\rput(16.494375,-1.16){\small check $Ram_2$ contents}
\psline[linewidth=0.04cm](16.5,0.5)(16.5,0.94)
\psline[linewidth=0.04cm](15.08,-1.86)(15.08,-1.4)
\psline[linewidth=0.04cm](10.88,-1.86)(10.88,-1.4)
\pspolygon[linewidth=0.04](0.0,-4.42)(0.4990476,-3.64)(20.460953,-3.64)(20.96,-4.42)
\usefont{T1}{ptm}{m}{n}
\rput(10.420468,4.57){sort by $Reg_k$ head position}
\psframe[linewidth=0.04,dimen=outer](4.0,6.8)(0.0,6.08)
\usefont{T1}{ptm}{m}{n}
\rput(2.0373437,6.455){\large $c_1$}
\psdots[dotsize=0.12](12.4,6.48)
\psdots[dotsize=0.12](13.0,6.48)
\psdots[dotsize=0.12](13.6,6.48)
\pstriangle[linewidth=0.04,dimen=outer](4.5,7.2)(4.6,1.96)
\usefont{T1}{ptm}{m}{n}
\rput(4.444375,7.46){\small check $Reg_k$ contents}
\psline[linewidth=0.04cm](4.5,9.12)(4.5,9.56)
\psframe[linewidth=0.04,dimen=outer](8.98,6.8)(4.98,6.08)
\usefont{T1}{ptm}{m}{n}
\rput(7.057344,6.455){\large $c_{5}$}
\psframe[linewidth=0.04,dimen=outer](20.98,6.8)(16.98,6.08)
\usefont{T1}{ptm}{m}{n}
\rput(19.167343,6.455){\large $c_{99}$}
\psline[linewidth=0.04cm](3.1,6.76)(3.1,7.22)
\psline[linewidth=0.04cm](5.9,6.76)(5.9,7.22)
\psline[linewidth=0.04cm](8.1,6.76)(8.1,7.22)
\psline[linewidth=0.04cm](17.88,6.76)(17.88,7.22)
\psline[linewidth=0.04cm](1.98,4.98)(1.98,6.1)
\psline[linewidth=0.04cm](6.98,4.98)(6.98,6.1)
\psline[linewidth=0.04cm](18.98,4.96)(18.98,6.12)
\pstriangle[linewidth=0.04,dimen=outer](9.5,7.2)(4.6,1.96)
\usefont{T1}{ptm}{m}{n}
\rput(9.444375,7.46){\small check $Reg_k$ contents}
\psline[linewidth=0.04cm](9.5,9.12)(9.5,9.56)
\pstriangle[linewidth=0.04,dimen=outer](16.5,7.2)(4.6,1.96)
\usefont{T1}{ptm}{m}{n}
\rput(16.444374,7.46){\small check $Reg_k$ contents}
\psline[linewidth=0.04cm](16.5,9.12)(16.5,9.56)
\psline[linewidth=0.04cm](15.08,6.76)(15.08,7.22)
\psline[linewidth=0.04cm](10.88,6.76)(10.88,7.22)
\pspolygon[linewidth=0.04](0.0,4.2)(0.4990476,4.98)(20.460953,4.98)(20.96,4.2)
\psline[linewidth=0.04cm](1.98,-1.84)(1.98,1.42)
\psline[linewidth=0.04cm](6.98,-1.84)(6.98,1.38)
\psline[linewidth=0.04cm](18.98,-1.84)(18.98,1.36)
\psline[linewidth=0.04cm](2.0,4.2)(2.0,3.16)
\psline[linewidth=0.04cm](6.98,4.2)(6.98,3.16)
\psline[linewidth=0.04cm](18.98,4.18)(18.98,3.14)
\psdots[dotsize=0.12,dotangle=-89.87178](10.778657,3.0399985)
\psdots[dotsize=0.12,dotangle=-89.87178](10.78,2.44)
\psdots[dotsize=0.12,dotangle=-89.87178](10.7813425,1.8400015)
\psframe[linewidth=0.04,dimen=outer](8.98,-7.62)(4.98,-8.34)
\pstriangle[linewidth=0.04,dimen=outer](23.38,11.0)(3.64,2.76)
\psline[linewidth=0.04cm](21.82,11.02)(21.82,9.56)
\psline[linewidth=0.04cm](24.98,11.02)(24.98,-10.66)
\psline[linewidth=0.04cm](24.36,11.0)(24.36,-4.86)
\psline[linewidth=0.04cm](23.78,11.0)(23.78,0.94)
\psline[linewidth=0.04cm](4.48,-10.64)(24.98,-10.64)
\psline[linewidth=0.04cm](4.48,-4.84)(24.34,-4.84)
\psline[linewidth=0.04cm](4.48,0.94)(23.8,0.94)
\psline[linewidth=0.04cm](4.48,9.58)(21.82,9.58)
\usefont{T1}{ptm}{m}{n}
\rput(23.36875,11.935){\large AND}
\psline[linewidth=0.04cm](23.38,13.7)(23.38,14.14)
\psdots[dotsize=0.12](22.2,10.24)
\psdots[dotsize=0.12](22.8,10.24)
\psdots[dotsize=0.12](23.4,10.24)
\end{pspicture} 
}
\caption{Each of the $T$ configurations has size $O(\log T)$.  The checking circuits
have size $\poly \log T$.  The sorting circuits have size $\tilde O(T)$.  $k$ is a constant.  Hence overall circuit has size $\tilde O(T)$.}
\label{fig:sort-and-check}
\end{figure}
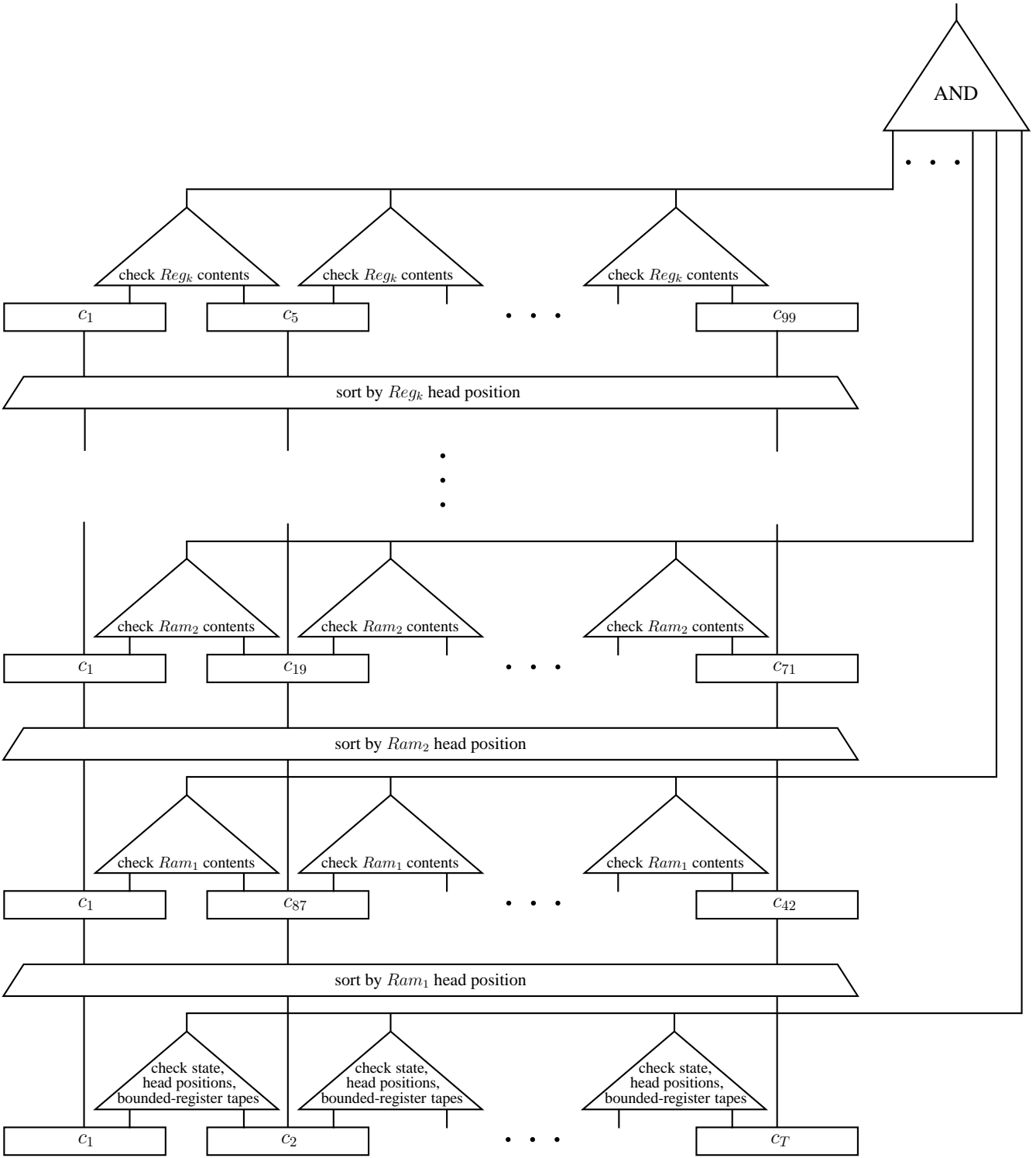


\paragraph{Making low-space computation local.}
Our first new idea is a general technique that we call \emph{spreading
computation}.  This shows that any circuit $C$
whose connections can be computed in
space linear in the description of a gate (i.e., space
log $|C|$) has an equivalent circuit $C'$ of size $|C'| =
\poly |C|$ whose connections can be
computed with constant locality.  This technique is
showcased in \S\ref{s-local-uniform} in the simpler
setting of Theorem \ref{t-local-uniformity}.

The main idea in the proof is
simply to let the gates of $C'$ represent configurations
of the low-space algorithm computing children in $C$. Then
computing a child amounts to performing one step of the
low-space algorithm, (each bit of) which can be done with
constant locality in a standard Turing machine model.
One complication with this approach is that the circuit
$C'$ has many invalid gates, i.e., gates that do
not correspond to the computation of the low-space
algorithm on a label of $C$.  This is necessarily so,
because constant locality is not powerful enough to even
check the validity of a configuration.  Conceivably,
these gates could induce loops that do not correspond to
computation, and make the final 3SAT instance always
unsatisfiable.  We avoid cycles by
including a clock in the configurations, which
allows us to ensure that each invalid
gate leads to a sink.

We apply spreading computation to the various sub-circuits checking
consistency of configurations, corresponding to the
triangles in Figure \ref{fig:sort-and-check}.  These
sub-circuits operate on configurations of size $O(\log T)$ and have size $\poly \log
T$.  Hence, we can tolerate the polynomial increase in
their complexity given by the spreading computation
technique.

There remain however tasks for which we cannot use
spreading computation.  One is the sorting sub-circuit.
Since it has size $> T$ we cannot afford a polynomial
increase.  Another task is indexing adjacent configurations.
We now discuss these two in turn.




\paragraph{Sorting.}
We first mention a natural approach that gets us close
but not quite to our main theorem.  The approach is to
define an appropriate labeling of the sorting network so
that its connections can be computed very efficiently. We
are able to define a labeling of bit-length $t + O(\log
t) = \log \tilde O(T)$ for comparators in the odd-even
mergesort network of size $\tilde O(2^t)$ (and depth
$t^2$) that sorts $T=2^t$ elements such that given a
label one can compute the labels of its children by a
decision tree of depth logarithmic in the length of the
label, i.e.~depth $\log \log \tilde O(T)$.  With a
similar labeling we can get linear size circuits.  Or we
can get constant locality at the price of making the 3SAT
instance of size $T^{1+\e}$.  The details appear in the
separate work \cite{JahanjouMV-sort}.

%
%

To obtain constant locality we use a variant by
Ben-Sasson, Chiesa, Genkin, and Tromer
\cite{Ben-SassonCGT12Fast}.  They replace sorting
networks with routing networks based on De Bruijn graphs.
We note that routing networks have been used extensively
in the PCP literature starting, to our knowledge, with
the work of Polishchuk and Spielman \cite{PolishchukS94}.
They have been used mostly for their algebraic
properties, whereas we exploit the small locality of
these networks.  Specifically, the connections of these
networks involve computing bit-shift, bit-xor, and
addition by 1. The first two operations can easily be
computed with constant locality, but the latter cannot in
the standard binary representation. However, this
addition by 1 is only on $O(\log \log T)$ bits.  Hence we
can afford an alternative, redundant representation which
gives us an equivalent network where all the operations
can be computed with constant locality.  This
representation again introduces invalid labels; those are
handled in a manner similar to our spreading computation
technique.

\paragraph{Plus one.}
Regardless of whether we are using sorting or routing
networks, another issue that comes up in all previous
proofs is addition by $1$ on strings of $> \log T$ bits.
This is needed to index adjacent configurations $C_i$ and
$C_{i+1}$ for the pairwise checks in Figure
\ref{fig:sort-and-check}. As mentioned before, this
operation cannot be performed with constant locality in
the standard representation. Also, we cannot afford a
redundant representation (since strings of length $c \log
T$ would correspond to an overall circuit of size $>
T^c$).

For context, we point out an alternative approach to
compute addition by 1 with constant locality which
however cannot be used because it requires an inefficient
pre-processing.  The approach is to use primitive
polynomials over GF(2)$^{\log T}$. These are polynomials
modulo which $x$ has order $2^{\log T} - 1$.  Addition by
1 can then be replaced by multiplication by $x$, which
can be shown to be local. This is similar to \emph{linear
feedback registers}. However, it is not known how to
construct such polynomials efficiently w.r.t.~their
degrees, see \cite{Shoup92}.

To solve this problem we use routing networks in a
different way from previous works.  Instead of letting
the network output an array $C_1, C_2, \ldots$
representing the sorted configurations, we use the
network to represent the ``next configuration'' map $C_i
\to C_{i+1}$.  Viewing the network as a matrix whose
first column is the input and the last column is the output, we
then perform the pairwise checks on every pair of input
and output configurations that are in the same row.
The bits of these configurations will be in the same
positions in the final label, thus circumventing addition
by one.

\medskip

As we mentioned earlier, for a result such as NEXP not in
ACC \cite{Williams10acc} it suffices to prove a weaker
version of our Theorem \ref{t-explicit-reductions} where
the reduction is computed by, say, an AC$^0$ circuit. For
the latter, it essentially suffices to show that either
the sorting network or the routing network's connections
are in that class.

\paragraph{Organization.}
\S\ref{s-local-uniform} showcases the spreading
computation technique and contains the proof of Theorem
\ref{t-local-uniformity}. In \S\ref{s:routing-networks}
we present our results on routing networks.  In
\S\ref{s:bit_fetching} we discuss how to fetch the bits
of the input $x$.  \S\ref{s-proof-of-the-main-theorem}
includes the proof of our main Theorem
\ref{t-explicit-reductions}.  Finally, \S\ref{s-sat-2-lb}
we discuss connections between satisfiability algorithms
and lower bounds; in particular we outline a proof of the
ACC lower bound, and prove Corollaries
\ref{co-sat-c-lb-c}, \ref{co-CNF-SAT-lb}, and
\ref{co-threshold-depth}.

\section{Spreading computation} \label{s-local-uniform}
In this section we prove Theorem \ref{t-local-uniformity}.

\tlocaluniformity*

We will use the following formalization of log-space
uniformity: a family of polynomial-size circuits $C' =
\{C'_n : \zo^n \to \zo\}_n$ is log-space uniform if there
exists a Turing machine $M$ that, on input $g \in
\zo^{\log |C'_n|}$ labeling a gate in $C'_n$, and $n$
written in binary, uses space $O(\log n)$ and outputs the
types and labels of each of $g$'s children. (Note that
$M$ outputs the types of $g$'s children rather than $g$'s
type; the reason for this will be clear from the
construction below.)

\begin{proof}
Let $C'$ be a log-space uniform family of polynomial-size
circuits and $M$ a log-space machine computing
connections in $C'$. We make the following simplifying
assumptions without loss of generality.
\begin{itemize}
\item Each gate in $C'$ has one of the following five types: And (fan-in-2), Not (fan-in-1), Input (fan-in-0), Constant-0 (fan-in-0), Constant-1 (fan-in-0).
\item For all $n$, $|C'_n|$ is a power of 2. In particular, each $(\log|C'_n|)$-bit string is a valid label of a gate in $C'_n$.
\item $M$'s input is a label $g \in \zo^{\log
    |C'_n|}$, a child-selection-bit $c \in \zo$ that
    specifies which of $g$'s $\leq 2$ children it
    should output, and $n$ in binary. $M$ terminates
    with $\log |C'_n|$ bits of its tape containing
    the child's label, and $3$ bits containing the
    child's type.
\end{itemize}

The local-uniform family $C$ will additionally have
fan-in-1 Copy gates that compute the identity function.
Gates in $C$ are labeled by configurations of $M$, and we
now specify these. Let $q,k,k' = O(1)$ be such that $M$
has $2^q-1$ states, and on input $(g,c,n) \in
\zo^{O(\log|C'_n|)}$ it uses space $\leq k \log n$ and
runs in time $\leq n^{k'}$. A configuration of $M$ is a
bit-string of length $((q+2) \cdot k + 2k') \cdot \log
n$, and contains two items: the {\em tape} and the {\em
timestep}.

The tape is specified with $(q+2) \cdot k \cdot \log n$
bits. Each group of $q+2$ bits specifies a single cell of
$M$'s tape as follows. The first two bits specify the
value of the cell, which is either $0$, $1$, or blank.
The remaining $q$ bits are all zero if $M$'s head is not
on this cell, and otherwise they contain the current
state of $M$.

The timestep is specified with $2k' \cdot \log n$ bits.
In order to allow it to be incremented by a local
function, we use the following representation which
explicitly specifies the carry bits arising from
addition. View the timestep as a sequence of pairs
$$((c_{k'\log n},b_{k'\log n}), (c_{k'\log n-1},b_{k'\log
n-1}), \ldots, (c_1, b_1)) \in \zo^{2k' \log n}.$$ Then
the timestep is initialized with $c_i = b_i = 0$ for all
$i$, and to increment by 1 we simultaneously set $c_1
\leftarrow b_1$, $b_1 \leftarrow b_1 \oplus 1$, and $c_i
\leftarrow b_i \wedge c_{i-1}$ and $b_i \leftarrow b_i
\oplus c_{i-1}$ for all $i > 1$.

It is not difficult to see that there is a local map $\mathsf{Upd} : \zo^{O(\log n)} \to \zo^{O(\log n)}$ that, on input a configuration of $M$, outputs the configuration that follows in a single step. Namely $\mathsf{Upd}$ increments the timestep using the method described above, and updates each cell of the tape by looking at the $O(1)$ bits representing that cell and the two adjacent cells.

We say that a configuration is {\em final} iff the
most-significant bit of the timestep is 1. This
convention allows a local function to check if a
configuration is final. Using the above method for
incrementing the timestep, a final configuration is
reached after $n^{k'} + k' \log n - 1$ steps. We say that
a configuration is {\em valid} if either (a) it is the
initial configuration of $M$ on some input $(g,c) \in
\zo^{\log|C'_n|+1}$ labeling a gate in $C'_n$ and
specifying one of its children, or (b) it is reachable
from such a configuration by repeatedly applying
$\mathsf{Upd}$. (Note that $\mathsf{Upd}$ must be defined
on every bit-string of the appropriate length. This
includes strings that are not valid configurations, and
on these it can be defined arbitrarily.)

\medskip

We now describe the circuit family $C = \{C_n\}_n$ and the local map $D$ that computes connections in these circuits, where $D$ depends on $n$.

$C_n$ has size $n^u$ for $u := (q+2)\cdot k + 2k' =
O(1)$, and each gate is labeled by an $(u \log n)$-bit
string which is parsed as a configuration of $M$. $C_n$
is constructed from $C'_n$ by introducing a chain of Copy
gates between each pair of connected gates
$(g_\mathsf{parent}, g_\mathsf{child})$ in $C'_n$, where
the gates in this chain are labeled by configurations
that encode the computation of $M$ on input
$g_\mathsf{parent}$ and with output $g_\mathsf{child}$.

Let $g \in \zo^{u \log n}$ be a configuration of $M$
labeling a gate in $C_n$. Our convention is that if $g$
is a final configuration then the type of $g$ is what is
specified by three bits at fixed locations on $M$'s tape,
and if $g$ is not a final configuration then the type is
Copy. (Recall that when $M$ terminates, the type of its
output is indeed written in three bits at fixed
locations.) In particular, the type of a gate can be
computed from its label $g$ by a local function.

$D$ computes the children of its input $g$ as follows. If
$g$ is not a final configuration, then $D$ outputs the
single child whose configuration follows in one step from
$g$ using the map $\textsf{Upd}$ described above. If $g$
is a final configuration, $D$ first determines its type
and then proceeds as follows. If the type is And, then
$D$ outputs two children by erasing all but the
$\log|C'_n|$ bits of $M$'s tape corresponding to a label
of a gate in $C'_n$, writing $n$, setting the timestep to
0, putting $M$ in its initial state with the head on the
leftmost cell, and finally setting one child to have $c =
0$ and one child to have $c = 1$. (Recall that $c$ is the
child-selection-bit for $M$.) If the type is Not, then
$D$ acts similarly but only outputs the one with $c=0$.
For any other type, $g$ has fan-in 0 and thus $D$ outputs
no children.

Naturally, the output gate of $C_n$ is the one labeled by
the configuration consisting of the first timestep whose
MSB is 1 and the tape of $M$ containing $(g_\mathsf{out},
t_\mathsf{out}, n)$ where $g_\mathsf{out}$ is the unique
label of $C'$'s output gate and $t_\mathsf{out}$ is its
type. (The remainder of this configuration can be set
arbitrarily.) It is clear that starting from this gate
and recursively computing all children down to the
fan-in-0 gates of $C_n$ gives a circuit that computes the
same function as $C'_n$. Call the tree computed in this
way the {\em valid tree}.

We observe that $C_n$ also contains gates outside of the
valid tree, namely all gates whose labels do not
correspond to a valid configuration. To conclude the
proof, we show that the topology of these extra gates
does not contain any cycles, and thus $C_n$ is a valid
circuit. By avoiding cycles, we ensure that the circuit
can be converted to a constraint-satisfaction problem
(i.e.\ 3SAT); the existence of a cycle with an odd number
of Not gates would cause the formula to be always
unsatisfiable.

Consider a label $g$ of a gate in $C_n$ containing a
configuration of $M$. If $g$ is the label of a gate in
the valid tree, then it is clearly not part of a cycle.
If $g$ is any other label, we consider two cases: either
$g$ is a final configuration or it is not. If $g$ is not
a final configuration, then its descendants eventually
lead to a final configuration $g'$. (This follows because
of the inclusion of the timestep in each configuration,
and the fact that starting from any setting of the
$(c_i,b_i)$ bits and repeatedly applying the increment
procedure will eventually yield a timestep with MSB $=
1$.)  Notice that the tape in $g'$ contains $\log|C'_n|$
bits corresponding to a valid label of a gate in $C'_n$.
(This is because of our convention that any bit-string of
that length is a valid label.  An alternative solution
intuitively connects the gates with MSB $=1$ to a sink,
but has other complications.) Therefore the children of
$g'$ are in the valid tree, and so $g'$ (and likewise
$g$) is not part of a cycle. Similarly, if $g$ is a final
configuration then its children are in the valid tree and
so it is not part of a cycle.
\end{proof}

\section{Routing networks} \label{s:routing-networks}

In this section we show how to non-deterministically implement the sorting subcircuits. We do this in a way so that for every input sequence of configurations, at least one non-deterministic choice results in a correctly sorted output sequence. Further, each possible output sequence either is a permutation of the input or contains at least one ``dummy configuration'' (wlog the all-zero string). Importantly, the latter case can be detected by the configuration-checking subcircuits.

\begin{theorem} \label{t:nondeterministic_DB_in_nc0}
Fix $T=T(n)\geq n$. Then for all $n>0$, there is a circuit
\[S:\left(\zo^{O(\log T)}\right)^T \times \zo^{O(T\log T)} \rightarrow \left(\zo^{O(\log T)}\right)^T \]
of size $T' := T \cdot \log^{O(1)} T$ and a labelling of the gates in $S$ by strings in $\zo^{\log T'}$ such that the following holds.
\begin{enumerate}
\item There is a local map $D : \zo^{\log T'} \times \zo \to \zo^{\log T'}$ such that for every label $g$ of a gate in $S$, $D(g, b)$ outputs the label of one of $g$'s $\leq 2$ children (according to $b$). Further, the type of each gate can be computed from its label in NC$^0$.  The latter NC$^0$ circuit is itself computable in time $\poly \log T$.
\item Given a $(\log T + O(\log\log T))$-bit index into $S$'s output, the label of the corresponding output gate can be computed in $\ncz$. Further, given any input gate label, the corresponding $(\log T + O(\log\log T))$-bit index into the input can be computed in $\ncz$.  These two NC$^0$ circuits are computable in time $\poly \log T$.
\item For every $C=(C_1, ..., C_T)$ and every permutation $\pi : [T] \to [T]$, there exists $z$ such that $S(C,z) = (C_{\pi(1)}, \ldots, C_{\pi(T)})$.
\item For every $C=(C_1, ..., C_T)$ and every $z$, if $(C'_1,\ldots,C'_T) := S(C,z)$ is not a permutation of the input then for some $i$, $C'_i$ is the all-zero string.
\end{enumerate}
\end{theorem}

We construct this circuit $S$ using a {\em routing network}.

\begin{definition}\label{d:rrnetwork}
Let $G$ be a directed layered graph with $\ell$ columns, $m$ rows, and edges only between subsequent columns such that each node in the first (resp.\ last) $\ell-1$ columns has exactly two outgoing (resp.\ incoming) edges.

$G$ is a {\em routing network} if for every permutation $\pi : [m] \to [m]$, there is a set of $m$ node-disjoint paths that link the $i$-th node in the first column to the $\pi(i)$-th node in the last column, for all $i$.
\end{definition}

Our circuit $S$ will be a routing network in which each node is a $2\times 2$ switch that either direct- or cross-connects (i.e.\ flips) its input pair of configurations to its output pair, depending on the value of an associated control bit.
This network is used to non-deterministically sort the input sequence by guessing the set of control bits. We use routing networks constructed from De Bruijn graphs as given in \cite{Ben-SassonCGT12Fast}.

\begin{definition} \label{d:debruijn}
An {\em $n$-dimensional De Bruijn graph} $DB_n$ is a directed layered graph with $n+1$ columns and $2^n$ rows. Each node is labeled by $(w, i)$ where $w\in \zo^n$ specifies the row and $0\leq i \leq n$ specifies the column. For $i  < n$, each node $(w,i)$ has outgoing edges to $(\mathsf{sr}(w),i+1)$ and $(\mathsf{sr}(w) \oplus 10\cdots 0, i+1)$, where $\mathsf{sr}$ denotes cyclic right shift.

A {\em $k$-tandem $n$-dimensional De Bruijn graph} $DB^k_n$ is a sequence of $k$ $n$-dimensional De Bruijn graphs connected in tandem.
\end{definition}

A proof of the following theorem can be found in
\cite{Ben-SassonCGT12Fast}.

\begin{theorem} \label{t:dbs_in_tandem}
For every $n$, $DB^4_n$ is a routing network.
\end{theorem}

To use this in constructing the sorting circuit $S$, we must show how, given the label $(w,i)$ of any node in $DB^4_n$, to compute in $\ncz$ the labels of its two predecessors.

Computing the row portion of each label (corresponding to $w$) is trivially an $\ncz$ operation, as $w$ is mapped to $\mathsf{sl}(w)$ and $\mathsf{sl}(w \oplus 10 \cdots 0)$ where $\mathsf{sl}$ denotes cyclic left shift.

For the column portion (corresponding to $i$), we use the
encoding of integers from Theorem
\ref{t-local-uniformity} that explicitly specifies the
carry bits arising from addition. Namely, we use a
$(2\log(4n))$-bit counter as described there, and number
the columns in reverse order so that the last column is
labeled by the initial value of the counter and the first
column is labeled by the maximum value.  This actually
results in more columns than needed, specifically $4n +
\log(4n)$, due to the convention that the counter reaches
its maximum when the MSB becomes 1. (We use this
convention here to determine when we are at the input
level of the circuit.) However, note that adding more
columns to $DB^4_n$ does not affect its rearrangeability
since whatever permutation is induced by the additional
columns can be accounted for by the rearrangeability of
the first $4n+1$ columns.

The next proof will introduce some dummy configurations
whose need we explain now.  With any routing network, one
can talk about either edge-disjoint routing or
node-disjoint routing.  Paraphrasing \cite[first par
after Def A.6]{Ben-SassonCGT12Fast}, a routing network
with $m$ rows can be used to route $2m$ configurations
using edge-disjoint paths (where each node receives and
sends two configs), or $m$ configurations using
node-disjoint paths (where each node receives and sends
one configuration).  In the former every edge carries a
configuration, while in the latter only half the edges
between each layer carry configurations (which half
depends on the permutation being routed).  However, when
implementing either type of routing with a boolean
circuit, all edges must of course always be present and
carry some data, because they correspond to wires. Thus
for node-disjoint routing, half of the edges between each
layer carry ``dummy configurations'' in our construction,
and it is possible even for the dummy configurations to
appear at the output of the network for certain (bad)
settings of the switches.  This whole issue would be
avoided with edge-disjoint routing (which is arguably
more natural when implementing routing networks with
circuits), but we prefer to rely on existing proofs as
much as possible.

\begin{proof}[Proof of Theorem \ref{t:nondeterministic_DB_in_nc0}]
The circuit $S$ is a De Bruijn graph with $T$ rows and $4 \log T + \log(4 \log T)$ columns as described above. It routes the $T$ configurations specified by its first input according to the set of paths specified by its second input. Each node not in the first or last column is a $2 \times 2$ switch on $O(\log T)$-bit configurations with an associated control bit from $S$'s second input specifying whether to swap the configurations. Each node in the last column has a control bit that selects which of its two inputs to output. Nodes in the first column map one input to two outputs; these have no control bits, and output their input along with the all-zero string.

We label each non-input gate in $S$ by $(t=00, w, i, s, d)$ where $t = 00$ specifies ``non-input'', $(w, i) \in \zo^{\log T} \times \zo^{O(\log\log T)}$ specifies a switch (i.e.\ a node in the De Bruijn graph), and $(s,d) \in \zo^{O(\log \log T)} \times \zo^{O(1)}$ specifies a gate within this switch. For the latter, we view a switch on $O(\log T)$-bit configurations as $O(\log T)$ switches on individual bits; then $s$ designates an $O(1)$-sized bit switch, and $d$ designates a gate within it.

We label each gate in $S$'s first input by $(t=01, w, s)$ where $t = 01$ specifies ``first input'', $w \in \zo^{\log T}$ specifies one of the $T$ configurations, and $s \in \zo^{O(\log \log T)}$ specifies a bit within the configuration.

We label each gate in $S$'s second input by $(t=10, w, i)$ where $t = 10$ specifies ``second input'' and $(w, i) \in \zo^{\log T} \times \zo^{O(\log\log T)}$ specifies a switch.

We take any gate with $t = 11$ to be a Constant-0 gate, one of which is used to output the all-zero string in the first column.

Naturally the labels of $S$'s output gates vary over $w$ and $s$ and have $t = 00$, $i = 0 \cdots 0$, and $d =$ the output gate of a bit switch; these and the input gate labels above give Property 2. Theorem \ref{t:dbs_in_tandem} guarantees that Property 3 holds for some setting of the switches, and it is straightforward to verify that Property 4 holds for any setting of the switches. We now show Property 1, namely how to compute connections in $S$ with a local map $D$.

Suppose $g = (t=00,w,i,s,d)$ is the label of a non-input gate, and let $b\in \zo$ select one of its children. There are four possible cases: (1) the child is in the same $2\times 2$ bit switch, (2) the child is an output gate of a preceding $2 \times 2$ bit switch, (3) the child is a bit of a configuration from $S$'s first input or the all-zero string, or (4) the child is a control bit from $S$'s second input. Since each bit switch has a fixed constant-size structure, the case can be determined by reading the $O(1)$ bits corresponding to $d$ and the MSB of $i$ which specifies whether $g$ is in the first column of the De Bruijn graph.

For case (1), $D$ updates $d$ to specify the relevant gate within the bit switch. For case (2), $D$ updates $w$ and $i$ via the procedures described above, and updates $d$ to specify the output gate of the new bit switch. For cases (3) and (4), $D$ updates $t$ and either copies the relevant portions ($w,s$ or $w,i$) from the rest of $g$ if $t \neq 11$, or sets the rest of $g$ to all zeros if $t = 11$.

Finally we note that as in Theorem \ref{t-local-uniformity}, there are strings that do not encode valid labels in the manner described above, and that these do not induce any cycles in the circuit $S$ due to the way the field $i$ is used.
\end{proof}

\section{Fetching bits}\label{s:bit_fetching}
In this section, we construct a local-uniform circuit
that will be used to fetch the bits of the fixed string
$x \in \zo^n$ in our final construction. Moreover, we
demonstrate a trade-off between the length of the labels
and the locality of the map between them.

\begin{theorem}\label{t:bit_fetching_tradeoff}
For all $x \in \zo^n$ and for all $r\in [n]$, there is a
circuit $C$ of size $2^\ell$ where $\ell = n/r + O(\log
n)$, a labeling of the gates in $C$ by strings in
$\zo^\ell$, and a map $D : \zo^\ell \to \zo^\ell$ each
bit of which is computable by a decision tree of depth
$O(\log r)$ with the following properties:
\begin{enumerate}
\item All gates in $C$ are either fan-in-0 Constant-0 or Constant-1 gates, or fan-in-1 Copy gates. In particular, $C$ has no input gates.
\item There are $n$ output gates $\mathsf{out}_1, \ldots, \mathsf{out}_n$, a Constant-0 gate $\mathsf{g}_0$, and a Constant-1 gate $\mathsf{g}_1$ such that, for all $i \leq n$, repeatedly applying $D$ to the label of $\mathsf{out}_i$ eventually yields the label of $\mathsf{g}_{x_i}$.
\item $\forall i \leq n$: the label of $\mathsf{out}_i$ can be computed in $\ncz$ from the binary representation of $i$.
\item Given $x$ and $r$ the decision trees computing
    $D$, and the $\ncz$ circuit in the previous item can be computed in time $\poly(n)$.
\end{enumerate}
\end{theorem}


\begin{proof}
We first explain the high-level idea of the construction. Assume $r=1$, we later extend this to the general case. For each $i$, the chain of labels induced by $D$ starting from the label of $\mathsf{out}_i$ will encode the following process. Initialize an $n$-bit string $s := 10 \cdots 0$ of Hamming weight 1. Then shift $s$ to the right $i-1$ times, bit-wise AND it with $x$, and finally shift it to the left $i-1$ times. Clearly, the leftmost bit of $s$ at the end of this process will equal $x_i$. The main technical difficulty is encoding counting to $i$ for arbitrary $i$ while allowing connections in $C$ to be computed locally. We achieve this using similar techniques as in the proof of Theorem \ref{t-local-uniformity} in Section \ref{s-local-uniform}, namely by performing the counting with a machine $M$ whose configurations we store in the labels of $C$. We now give the details.

The label of each gate in $C$ is parsed as a tuple $(t,s,d,i,c)$ of length $\ell = n + O(\log n)$ as follows: $t$ is a $2$-bit string specifying the type of the gate, $s$ is the $n$-bit string described above, $d$ is a $1$-bit flag specifying the direction $s$ is currently being shifted (left or right), $i$ is the $(\log n)$-bit binary representation of an index into $x$, and $c$ is the $O(\log n)$-bit configuration of a machine $M$ that operates as follows. $M$ has $\log n$ tape cells initialized to some binary number. It decrements the number on its tape, moves its head to the left-most cell and enters a special state $q^*$, and then repeats this process, terminating when its tape contains the binary representation of $1$. We encode $M$'s $O(\log n)$-bit configurations as in Theorem \ref{t-local-uniformity}, in particular using the same timestep format so that checking if $M$ has terminated can be done by reading a single bit.

The label of $\mathsf{out}_i$ has the following natural form. The type $t$ is Copy, $s$ is initialized to $10 \cdots 0$, the flag $d$ encodes ``moving right'', $i$ is the correct binary representation, and $c$ is the initial configuration of $M$ on input $i$. Note that this can be computed from $i$ in $\ncz$.

The local map $D$ simply advances the configuration $c$, and shifts $s$ in the direction specified by $d$ iff it sees the state $q^*$ in $M$'s left-most cell. If $c$ is a final configuration and $d$ specifies ``moving right'', then $D$ bit-wise ANDs $x$ to $s$, sets $d$ to ``moving left'', and returns $M$ to its initial configuration on input $i$. If $c$ is a final configuration and $d$ specifies ``moving left'', then $D$ outputs the unique label of the constant gate $\mathsf{g}_b$ where $b$ is the left-most bit of $s$. (Without loss of generality, we can take this to be the label with the correct type field and all other bits set to 0.)

The correctness of this construction is immediate. Furthermore, the strings that encode invalid labels do not induce cycles in $C$ for similar reasons as those given at the end of Theorem \ref{t-local-uniformity}. (In fact, the presence of cycles in this component would not affect the satisfiability of our final 3SAT instance, since the only gates with non-zero fan-in have type Copy.)

We now generalize the proof to any value of $r$. The goal is to establish a trade-off between the label length $\ell$ and the locality of the map $D$ such that at one extreme we have $\ell=n+O(\log n)$ and $D$ of constant locality and the the other we have $\ell=O(\log n)$ and $D$ computable by decision trees of depth $O(\log n)$.

The construction is the same as before but this time the label of a gate in $C$ is parsed as a tuple $(t, p, k, d, i, c)$ of length $\ell = n/r + \log r + O(\log n) = n/r + O(\log n)$, where $t, d, i$, and $c$ are as before and $p \in \zo^{n/r}$ and $k \in \zo^{\log r}$ together represent a binary string of length $n$ and Hamming weight 1. More precisely, consider a binary string $s\in \zo^n$  of Hamming weight 1 partitioned into $r$ segments each of $n/r$ bits. Now, the position of the bit set to 1 can be determined by a segment number $k \in \zo^{\log r}$ and a bit string $p \in \zo^{n/r}$ of Hamming weight 1 within the segment.

The map $D$ now {\em cyclically} shifts the string $p$ in
the direction indicated by $d$, updating $k$ as needed.
For the rest, the behavior of $D$ remains unchanged. In
particular, if $c$ is a final configuration and $d$
specifies ``moving right'', then $D$ bit-wise ANDs the
relevant $n/r$-bit segment of $x$ to $p$ and so on. To
perform one such step, $D$ needs to read the entire $k$
in addition to a constant number of other bits, so it can
be computed by decision trees of depth $O(\log r)$.
\end{proof}

\section{Putting it together} \label{s-proof-of-the-main-theorem}

We now put these pieces together to prove Theorem
\ref{t-explicit-reductions}.  First we modify previous
proofs to obtain the following normal form for
non-deterministic computation that is convenient for our
purposes, cf.~\S\ref{s-intro-techniques}.

\begin{theorem}\label{thm:GS}
Let $M$ be an algorithm running in time $T = T(n) \geq n$
on inputs of the form $(x,y)$ where $|x| = n$. Then there
is a function $T' = T \log^{O(1)} T$, a constant $k =
O(1)$, and $k$ logspace-uniform circuit families
$C_1,\ldots,C_k$ each of size $\log^{O(1)} T$ with oracle
access to $x$, such that the following holds:

For every $x \in \zo^n$, there exists $y$ such that
$M(x,y)$ accepts in $\leq T$ steps iff there exists a
tuple $(z_1,\ldots,z_{T'}) \in \left(\zo^{O(\log
T)}\right)^{T'}$, and $k$ permutations
$\pi_1,\ldots,\pi_k : [T'] \to [T']$ such that for all $j
\leq k$ and $i \leq T'$, $C_j\left(z_i, z_{\pi_j(i)}
\right)$ outputs $1$.
\end{theorem}

We note that ``oracle access to $x$'' means that the
circuits have special gates with $\log n$ input wires
that output $x_i$ on input $i \leq n$ represented in
binary. Alternatively the circuits $C_i$ do not have
oracle access to $x$ but instead there is a separate
constraint that, say, the first bit of $z_i$ equals $x_i$
for every $i \le n$.

\begin{proof}[Proof sketch]
Model $M$ as a random-access Turing machine running in
time $T'$ and using indices of $O(\log T') = O(\log T)$
bits.  All standard models of computation can be
simulated by such machines with only a polylogarithmic
factor $T'/T$ blow-up in time.  Each $z_i$ is an $O(\log
T)$-bit configuration of $M$ on some input $(x,y)$. This
configuration contains the timestamp $i \leq T'$, the
current state of $M$, the indices, and the contents of
the indexed memory locations; see \cite{ViolaNEU-ram2sat}
for details.

The circuits and permutations are used to check that
$(z_1,\ldots,z_{T'})$ encodes a valid, accepting
computation of $M(x,y)$. This is done in $k+1$ phases
where $k = O(1)$ is the number of tapes. First, we use
$C_1$ to check that each configuration $z_i$ yields
$z_{i+1}$ assuming that all bits read from memory are
correct, and to check that configuration $z_{T'}$ is
accepting. (For this we use the permutation $\pi_1(i) :=
i + 1\, \bmod T'$.) This check verifies that the state,
timestamp, and indices are updated correctly. To
facilitate the subsequent checks, we assume without loss
of generality that $M$'s first $n$ steps are a pass over
its input $x$. Therefore, $C_1$ also checks (using oracle
access to $x$) that if the timestamp $i$ is $\leq n$ then
the first index has value $i$ and the bit read from
memory is equal to $x_i$.

For $j > 1$, we use $C_j$ to verify the correctness of
the read/write operations in the $(j-1)$-th tape.  To do
this, we use the permutation $\pi_j$ such that for each
$i$, $z_i$ immediately precedes $z_{\pi_j(i)}$ in the
sequence of configurations that are sorted first by the
$(j-1)$-th index and then by timestamp. Then, $C_j$ checks
that its two configurations are correctly sorted, and
that if index $j-1$ has the same value in both then the
bit read from memory in the second is consistent with the
first.
It also checks that the value of any location that
is read for the first time is blank, except for the portion on the first tape that corresponds to the input $(x,y)$.
(Note that $C_1$ already verified that the first time $M$ reads a memory
index $i \leq n$, it contains $x_i$.  No checks is performed on the $y$ part, corresponding to this string being existentially quantified.)

We stipulate that each $C_j$ above outputs $0$ if either of its
inputs is the all-zero string, which happens if the sorting circuit does not
produce a permutation of the configurations (cf.\ Theorem
\ref{t:nondeterministic_DB_in_nc0}, part 4).
Finally, we observe that all checks can be
implemented by a log-space uniform family of
polynomial-size circuits with oracle access to $x$.
\end{proof}

We now prove our main theorem, restated for convenience.
The high-level idea is to use
\S\ref{s-local-uniform}-\ref{s:bit_fetching} to transform
the circuits from Theorem \ref{thm:GS} into circuits
whose connections are computable by small-depth decision
trees, and to then apply the textbook reduction from
Circuit-SAT to 3SAT.

\texplicitreductions*

\begin{proof}
We parse $D$'s input as a tuple $(g,r,s)$, where $g$ is
the label of a gate in some component from Theorem
\ref{thm:GS}, as explained next, $r$ is a 2-bit clause
index, and $s$ is a 1-bit control string. We specifically
parse $g$ as a pair (Region, Label) as follows.  Region
(hereafter, $R$) is an $O(1)$-bit field specifying that
Label is the label of either
\begin{enumerate}[(a)]
\item a gate in a circuit that implements the $i$th instance of some $C_j$,
\item a gate in a circuit that provides oracle access to $x$,
\item a gate in a circuit that implements some
    $\pi_j$ via a routing network, or
\item a gate providing a bit of some configuration $z_i$.
\end{enumerate}


Label (hereafter, $L$) is a $(\max(\log T, n/r) + O(\log
n) + O(\log\log T))$-bit field whose interpretation
varies based on $R$. For (a), we take $L = (i,j,\ell)$
where $i \leq T$ and $j \leq k$ specify
$C_j(z_i,z_{\pi_j(i)})$ and $\ell \in \zo^{O(\log\log
T)}$ specifies a gate within it, where we use Theorem
\ref{t-local-uniformity} and take $C_j$ to be a circuit
whose connections are computable in $\ncz$.  For (b), we
take $L$ to be a $(n/r + O(\log n))$-bit label of the
circuit from Theorem \ref{t:bit_fetching_tradeoff}. For
(c), we take $L = (j,\ell)$ where $j \leq k$ specifies
$\pi_j$ and $\ell \in \zo^{\log T + O(\log\log T)}$
specifies a gate in the circuit from Theorem
\ref{t:nondeterministic_DB_in_nc0} implementing $\pi_j$.
For (d), $L$ is simply the $(\log T + O(\log\log T))$-bit
index of the bit.

We now describe $D$'s computation. First note that from
Theorems \ref{t-local-uniformity},
\ref{t:nondeterministic_DB_in_nc0}, and
\ref{t:bit_fetching_tradeoff}, the type of $g$ can be
computed from $L$ in $\ncz$; call this value Type $\in
\{$And, Not, Copy, Input, $x$-Oracle, Constant-0,
Constant-1$\}$.

\paragraph{Computing $g$'s children.}
$D$ first computes the labels of the $\leq 2$ children of
the gate $g=(R,L)$ as follows.

If $R$ specifies that $L = (i,j,\ell)$ is the label of a
gate in $C_j(z_i,z_{\pi_j(i)})$, $D$ computes $\ell$'s
child(ren) using the $\ncz$ circuit given by Theorem
\ref{t-local-uniformity}. The only cases not handled by
this are when Type $\in \{x$-Oracle, Input$\}$. When Type
$= x$-Oracle, the child is the $i'$th output gate of the
bit-fetching circuit, where $i'$ is the lower $\log n$
bits of $i$; by part 3 of Theorem
\ref{t:bit_fetching_tradeoff}, the label of this gate can
be computed in $\ncz$. When Type = Input, the child is
either the $m$th bit of $z_i$ or the $m$th bit of
$\pi_j$'s $i$th output, for some $m \leq O(\log T)$. We
assume without loss of generality that $m$ is contained
in binary in a fixed position in $L$, and that which of
the two inputs is selected can be determined by reading a
single bit of $L$. Then, the label of the bit of $z_i$
can be computed in $\ncz$ by concatenating $i$ and $m$,
and the label of the $m$th bit of $\pi_j$'s $i$th output
can be computed by part 2 of Theorem
\ref{t:nondeterministic_DB_in_nc0}.

If $R$ specifies that $L$ is a label in the bit-fetching
circuit from Theorem \ref{t:bit_fetching_tradeoff}, $D$
computes its child using the $O(\log r)$-depth decision
trees given by that theorem.

If $R$ specifies that $L = (j,\ell)$ is the label of a
sorting circuit from Theorem
\ref{t:nondeterministic_DB_in_nc0}, $D$ computes $\ell$'s
child(ren) using the $\ncz$ circuit given by that
theorem. The only case not handled by this is when $\ell$
labels a gate in the first input to the sorting circuit,
but in this case the child is a bit of some $z_i$ where
$i$ can be computed in $\ncz$ by part 2 of Theorem
\ref{t:nondeterministic_DB_in_nc0}.

If Type = Input and $(R,L)$ is not one of the cases mentioned above or Type $\in \{$Constant-0, Constant-1$\}$, $D$ computes no children.

\paragraph{Outputting the clause.}
When the control string $s = 0$, $D$ outputs the clause specified by $g$ and $r$ in the
classical reduction to 3SAT, which we review now. (Recall that $r$ is a 2-bit clause index.) The
3SAT formula $\phi$ contains a variable for each gate $g$, including each input
gate, and the clauses are constructed as follows.

If Type = And, we denote $g$'s children by $g_a$ and $g_b$. Then depending on the
value of $r$, $D$ outputs one of the four clauses in the
formula
\[(g_a \vee g_b \vee \overline{g}) \wedge
  (g_a \vee \overline{g_b} \vee \overline{g}) \wedge
  (\overline{g_a} \vee g_b \vee \overline{g}) \wedge
  (\overline{g_a} \vee \overline{g_b} \vee g). \]
These ensure that in any satisfying assignment, $g =
g_a \wedge g_b$.

If Type = Not, we denote $g$'s child by $g_a$. Then depending on the value of $r$, $D$ outputs
one of the two clauses in the formula
\[(g \vee g_a \vee g_a)\wedge
  (\overline{g} \vee \overline{g_a}\vee \overline{g_a}).\]
These ensure that in any satisfying assignment, $g =
\overline{g_a}$.

If Type $\in \{x$-Oracle, Copy$\}$ or Type = Input and $D$ computed $g$'s child $g_a$, then depending on the value of $r$, $D$ outputs
one of the two clauses in the formula
\[(\overline{g} \vee g_a \vee g_a)\wedge
  (g \vee \overline{g_a}\vee \overline{g_a}).\]
These ensure that in any satisfying assignment, $g =
g_a$.

If Type = Constant-0, $D$ outputs
the clause $(\overline{g} \vee \overline{g} \vee
\overline{g})$ which ensures that in any satisfying
assignment $g$ is false (i.e.\ that each Constant-0
gate outputs 0). If Type = Constant-1, $D$ outputs the clause
$(g \vee g \vee g)$ which ensures that in
any satisfying assignment $g$ is true (i.e.\ that
each Constant-1 gate outputs 1).

If Type = Input, and $D$ did not compute a child of $g$, $D$ outputs a dummy clause
$(g_\mathsf{dummy} \vee g_\mathsf{dummy} \vee
g_\mathsf{dummy})$ where $g_\mathsf{dummy}$ is a
string that is distinct from all other labels $g$.

\smallskip

When the control string $s = 1$, $D$ outputs clauses
encoding the restriction that each
$C_j(z_i,z_{\pi_j(i)})$ outputs 1. Namely, $D$ parses $L
= (i,j,\ell)$ as above, and outputs $(g_{i,j} \vee
g_{i,j} \vee g_{i,j})$, where $g_{i,j} := (i,j,\ell^*)$
and $\ell^*$ is the label of $C_j$'s output gate, which
depends only on $j$ and $\log T$ and thus can be
hardwired into $D$.
\end{proof}

\section{From satisfiability to lower bounds, tightly}
\label{s-sat-2-lb}

In this section we first outline a proof of the ACC lower
bound.  We then prove Corollaries \ref{co-sat-c-lb-c},
\ref{co-CNF-SAT-lb}, and \ref{co-threshold-depth}.  For
simplicity throughout this section we focus on lower
bounds for computing functions in the class
E$^{\mathrm{NP}}$.  As remarked earlier, in a few cases
the class can be improved to NE using
\cite{ImpagliazzoKaWi01,Williams13-Improving,Williams10acc}.

\begin{theorem}[\cite{Williams10acc}] \label{t-acc}
{\em E}$^\mathrm{NP}$ does not have polynomial-size {\em ACC}
circuits.
\end{theorem}
\begin{proof}[Proof sketch]
Following \cite{Williams13-Improving}, suppose that
E$^\mathrm{NP}$ has ACC circuits of size $n^c$ and depth
$d$ for some constants $c$ and $d$. Let $L \in
\mathrm{NTime}(2^n) \setminus \mathrm{NTime}(o(2^n))$
\cite{Cook73,SeiferasFM78,Zak83}. Consider the
E$^\mathrm{NP}$ algorithm that on input $x \in \zo^n$ and
$i \le 2^n \poly(n)$ computes the 3CNF $\phi_x$ from
Theorem \ref{t-explicit-reductions}, computes its first
satisfying assignment if one exists, and outputs its
$i$th bit.  By assumption this algorithm can be computed
by an ACC circuit of depth $d$ and size $n^c$.  By
hardwiring $x$ we obtain that for every $x \in \zo^n$
there is a circuit $C_x$ of the same depth and size that
on input $i$ computes that $i$th bit.

We contradict the assumption on $L$ by showing how to
decide it in $\mathrm{Ntime}(o(2^n))$.  Consider the
algorithm that on input $x \in \zo^n$ guesses the above
circuit $C_x$.  Then it connects three copies of $C_x$ to
the decision trees with depth $O(1)$ from Theorem
\ref{t-explicit-reductions} that on input $j \in \zo^{n +
O(\log n)}$ compute the $j$th clause of $\phi_x$, to
obtain circuit $C'_x$.  Since the paths in a decision
tree are mutually exclusive, $C'_x$ may be obtained
simply by appending an $n^{O(1)}$-size layer of And gates
to a layer of the gates of $C_x$, and increasing the
fan-in of the latter, for a total depth of $d+1$.  Then
the algorithm constructs the circuit $C''_x$ which in
addition checks if the outputs of the 3 copies of $C_x$
do not satisfy the $j$th clause.  The circuit $C''_x$ is
ACC, and a naive implementation yields size $n^{O(c)}$
and depth $d+3$.  Running the satisfiability algorithm in
\cite{Williams10acc} on $C''_x$ determines if $\phi_x$ is
satisfiable, and hence if $x \in L$, in time
$2^{|j|}/|j|^{\omega(1)} = 2^{n+O(\log n)}/(n+O(\log
n))^{\omega(1)} = o(2^n)$.
\end{proof}

\begin{proof}[Proof of Corollary \ref{co-sat-c-lb-c}]
We reason as in the proof of Theorem \ref{t-acc}. For
(1)-(3) it is immediate to see that the circuit $C''_x$
is in the same class as $C_x$.  For (4) recall from
\cite{BeT94} that quasipolynomial-size SYM-AND circuits
are equivalent to functions of the form
$f(p(x_1,\ldots,x_n))$, where $p$ is a polylog-degree
polynomial with integer coefficients of quasipolynomial
magnitude, and $f$ is any boolean function.  Our
reduction can be implemented by polynomials of constant
degree, which can be absorbed in $p$.  There remains to
negate the outputs depending on the negation bits, and to
take an Or.  This can be viewed as having $O(1)$
functions of the form $f_i(p_i(x_1,\ldots,x_n))$ and
wanting some arbitrary function of their joint output.
Let $M$ be the smallest number which is bigger than the
absolute value attained by any $p_i$.  $M$ can be taken
to be quasipolynomial.  Now define the polynomial $p :=
\sum_{i} M^i p_i$.  The output $p$ encodes the outputs of
the $p_i$, and so we conclude by applying an appropriate
function $f$ to $p$.
\end{proof}

\begin{proof}[Proof of Corollary \ref{co-CNF-SAT-lb}]
(1') is just the Cook-Levin theorem.

(3') follows from Valiant's depth reduction
\cite{Val77}, an exposition of which is given in \cite{viola-FTTCS09}.
It shows how to convert any linear-size log-depth circuit
into an Or of $2^{O(n/\log \log n)}$ $k$CNFs for $k = n^\e$
for any fixed $\e$.  (This transformation requires
selecting $O(n/ \log \log n)$ wires to be removed from the
circuit so that the depth is reduced.  It can be verified
that these wires can be computed efficiently so that the
whole transformation can be carried out in polynomial
time in the output circuit; in fact an efficient
co-nondeterministic procedure suffices for the argument,
so one could also co-nondeterministically guess the wires
and verify that the depth is reduced.)  Thus, if we can
determine the satisfiability of any such $k$CNF in time
$2^{n - \omega(n/\log \log n)}$, we can determine the
satisfiability of the original circuit in time
$2^{O(n/\log \log n) + n - \omega(n/\log \log n)} = 2^{n -
\omega(n/\log \log n)} = 2^n/n^{\omega(1)}$.

(2') follows from Valiant's depth reduction for series-parallel circuits \cite{Val77}, some missing details of which are provided by Calabro \cite{Calabro08} (see remarks following this proof). It shows how to convert any linear-size series-parallel circuit into an Or of $2^{n/d}$ $k$CNFs for $k = 2^{2^{O(d)}}$, for any parameter $d$ (cf. \cite[Thm.\ 3.9]{CyganDLMNOPSW-onpro}). Thus similarly to (3'), if we can determine the satisfiability of any such $k$CNF in time $2^{\eps n}$ for some fixed $\eps < 1$, then choosing $d = 2/(1-\eps) = O(1)$ we can determine the satisfiability of the original circuit in time $2^{n/d} \cdot 2^{\eps n} = 2^{(1+\eps)n/2} = 2^n/n^{\omega(1)}$.
\end{proof}

\paragraph{Remarks on series-parallel circuits.}  The
definition we work with is given by Calabro  building on
Valiant's seminal work \cite{Val77}, see the definition
of ``VSP'' in \cite{Calabro08}.  It has been noted
several times that such circuits of linear size may be
transformed into subexponential Or-And-Or circuits with
bounded bottom fan-in, see
e.g.~\cite{CyganDLMNOPSW-onpro}. However, a proof seems
missing from the literature.  It is not immediate to
obtain this proof from Theorem 8 in \cite{Calabro08},
because the statement of that theorem talks about long
input-output paths, whereas for the transformation one
needs to forbid any long path.  With the latter the proof
is completed reasoning as in \cite{Val77,viola-FTTCS09}.
Nevertheless, with essentially the same proof as that of
Theorem 8 in \cite{Calabro08} one can forbid any long
path.

We also remark that several papers add the restriction
that the series-parallel circuit be of logarithmic depth.
This restriction is unnecessary.

%

\begin{proof}[Proof of Corollary \ref{co-threshold-depth}]
We reason as in the proof of Theorem \ref{t-acc}, and
improve the depth of $C''_x$ to $d+2$ as follows.  The
algorithm will guess instead of $C_x$ a circuit $D_x$
that given $i$ and a bit $b$ computes the $i$th bit
mentioned above xor'ed with $b$.  In an additional
preliminary stage, the algorithm will check the
consistency of $D$ by running the satisfiability
algorithm on $(D_x(i,0) \wedge D_x(i,1))$ and on
$(\neg D_x(i,0)) \wedge \neg D_x(i,1))$, and reject if
the output is ever 1.  This circuit can be implemented in
depth $d+1$.
\end{proof}

\begin{proof}[Proof sketch of Corollary \ref{co-3sat-with-3col}]
Fix any language $L \in \mathrm{NTime}(2^n) \setminus \mathrm{NTime}(o(2^n))$. Theorem 2.2 of \cite{Ben-SassonV-AC0PCP} shows the following: given $x \in \zo^n$, one can produce in time $\poly(n)$ a map $D$ such that $D(i)$ outputs the $i$-th clause in a 3CNF of size $\widetilde{O}(2^n)$ that is satisfiable iff $x \in L$, and further $D$ is computable with locality 1 after fixing some $O(\log n)$ bits of its input.

Here, we compose this with the textbook reduction from 3SAT to 3Color. Recall that this reduction has a pair of nodes for each variable and a constant-size gadget for each clause. Moreover, there are no edges between clause gadgets (see remarks following this proof). Thus we can instead produce a map $D'$ where $D'(i)$ outputs the $i$-th edge in a graph of size $\widetilde{O}(2^n)$ that is 3-colorable iff $x \in L$, and $D'$ is computable with locality 1 after fixing some $O(\log n)$ bits of its input. This is done as follows. We parse a typical input to $D'$ as an index to a clause (i.e.\ an input to $D$) plus $O(1)$ control bits specifying which edge in the clause gadget to output. We include the control bits in the $O(\log n)$ bits to be fixed. Once these are fixed, the labels of the two vertices output by $D'$ are just projections of the control bits, the clause index, and the variables in the clause. The latter are output by $D$, so we can modify $D$ to obtain $D'$ that is 1-local after $O(\log n)$ bits are fixed.

With this modification to the proofs of lower bounds from satisifiability algorithms (e.g.\ Theorem \ref{t-acc} above), we now need only two (instead of three) copies of the circuit implementing the E$^{NP}$ search algorithm since we have a 2CSP instance (3Color) rather than a 3CSP instance (3SAT). On the other hand, the values are no longer Boolean; in the case of 3Color, three values are possible which we encode with two bits.
Now in order to derive a lower bound against circuits with $3n$ non-input gates, one needs to decide in time $2^n/n^{\omega(1)}$ the satisfiability of a circuit with $2 \cdot 3n$ non-input gates, $n$ input gates, and $O(1)$ other gates. The Cook-Levin theorem reduces this to a 3SAT instance on $7n + O(1)$ variables, so it suffices to solve 3SAT in deterministic time $c^n$ for  $c<2^{1/7}$. The remainder of the proof follows \cite[Thm.\ 1.5]{Ben-SassonV-AC0PCP}.
\end{proof}

In the above proof, a crucial feature of the reduction from 3SAT (or CircuitSAT) to 3Color is that there are no edges between gadgets corresponding to different clauses (or gates). Consider instead the standard reduction from 3SAT to Independent Set, which likewise produces a constant-size gadget for each clause but which also contains edges between gadgets whose clauses have conflicting literals (i.e.\ $u$ and $\overline{u}$). In this case it is not clear how to construct an $\ncz$ map $D$ that outputs edges in the graph. This is because outputting the inter-clause edges seems to require computing, from the label of a gate $g$ in the CircuitSAT instance, the labels of all gates that share a common input with $g$ (such gates induce clauses with conflicting literals, cf.\ proof of Theorem \ref{t-explicit-reductions}).

In light of this it is interesting to compare our results with \cite{AroraSW09}. There the authors define a $\mathcal{C}$-graph as one for which there exists a map $D \in \mathcal{C}$ such that $D(i,j)$ outputs the $(i,j)$-th bit of the graph's adjacency matrix. They show that if NEXP $\neq$ EXP then none of the standard NP-complete graph problems (including 3Color and Independent Set) can be solved on $\acz$-graphs in time polynomial in the size of the graph \cite[Thm.\ 5]{AroraSW09}. One can instead consider a definition that corresponds to our construction above: $G$ is a $\mathcal{C}$-graph if there exists a map $D \in \mathcal{C}$ such that $\{D(i)\, |\, i\}$ is the edge set of $G$. Using this latter definition, our results imply that if NEXP $\neq$ EXP then 3Color cannot be solved on $\ncz$-graphs in polynomial time. However as just discussed, it is not clear how to prove such a result for e.g.\ Independent Set.

%
%

\paragraph{Acknowledgments.}  We are very grateful to Eli
Ben-Sasson for a discussion on routing networks which led
us to improving our main result,
cf.~\S\ref{s-intro-techniques}.  We also thank Ryan
Williams for feedback on the write-up.

\bibliographystyle{alpha}
\small{ \ifnum\EmanueleViolaDir=1
\bibliography{../OmniBib}
\else
\bibliography{OmniBib}
\fi }

\end{document}